\documentclass[12pt]{iopart}

\usepackage{amssymb}
\usepackage{amsthm}
\usepackage{amsbsy}
\usepackage{graphicx}
\usepackage{fullpage}
\usepackage{iopams}

\newtheorem{theorem}{Theorem}
\newtheorem{lemma}{Lemma}

\theoremstyle{definition}
\newtheorem{remark}{Remark}

\newcommand{\eps}{\varepsilon}
\newcommand\J{\mathbf{J}}
\newcommand\U{\mathbf{U}}

\newcommand\I{\mathbf{I}}
\renewcommand\L{\mathbf{L}}
\newcommand\LL{\mathbf{\Lambda}}
\newcommand\Paths{\mathbb{P}}
\newcommand\path{\mathbf{p}}

\newcommand{\vS}{\boldsymbol{\sigma}}
\newcommand{\bS}{{\mathbf S}}

\newcommand\m{\mathbf{m}}
\newcommand\R{\mathbf{R}}
\newcommand\vA{\mathbf{A}}
\newcommand\gf{\boldsymbol{\phi}}

\newcommand{\ui}{\textrm{i}}
\newcommand{\ue}{\textrm{e}}

\newcommand{\ud}{\mathrm{d}}

\newcommand{\A}{{\mathbb A}}
\newcommand{\B}{{\mathbb B}}

\newcommand{\lmin}{L_\mathrm{min}}

\newcommand{\PP}{\mathcal{P}}
\newcommand{\cL}{\mathcal{L}}
\newcommand{\cV}{\mathcal{V}}
\newcommand{\cB}{\mathcal{B}}

\newcommand{\al}{\alpha}

\newcommand{\alb}{\bar{\alpha}}

\newcommand{\bof}{\mathbf{f}}
\newcommand{\bz}{\mathbf{0}}
\newcommand{\ch}{\hat{c}}
\newcommand{\bog}{\mathbf{g}}

\newcommand{\re}{\mathrm{Re}}
\newcommand{\im}{\mathrm{Im}}

\begin{document}

\title{Mathematical aspects of vacuum energy on quantum graphs}

\author{G Berkolaiko$^1$, J M Harrison$^{1,2}$ and J H Wilson$^{1,3}$}\address{$^1$ Dept. of Mathematics, Texas A\&M University, College Station, TX
  77843-3368, USA}\address{$^2$ Department of Mathematics, Baylor University, Waco, TX 76798,
    USA}\address{$^3$ University of Maryland, Department of Physics, College Park, MD 20742, USA}
\ead{\mailto{gregory.berkolaiko@math.tamu.edu}, \mailto{jon\_harrison@baylor.edu}, \mailto{jwilson.thequark@gmail.com}}
\date{}


\begin{abstract}
  We use quantum graphs as a model to study various mathematical
  aspects of the vacuum energy, such as convergence of periodic path
  expansions, consistency among different methods (trace formulae
  versus method of images) and the possible connection with the underlying
  classical dynamics.
  In our study we derive an expansion for the vacuum energy in
  terms of periodic paths on the graph and prove its convergence and
  smooth dependence on the bond lengths of the graph.  For an
  important special case of graphs with equal bond lengths, we derive a
  simpler explicit formula.  With minor changes this formula also
  applies to graphs with rational (up to a common factor) bond
  lengths.
  The main results are derived using the trace formula.  We also
  discuss an alternative approach using the method of images and prove
  that the results are consistent.  This may have important
  consequences for other systems, since the method of images, unlike
  the trace formula, includes a sum over special ``bounce paths''.  We
  succeed in showing that in our model bounce paths do not contribute
  to the vacuum energy.
  Finally, we discuss the proposed possible link between the magnitude
  of the vacuum energy and the type (chaotic vs.\ integrable) of the
  underlying classical dynamics.  Within a random matrix model we
  calculate the variance of the vacuum energy over several ensembles
  and find evidence that the level repulsion leads to suppression of
  the vacuum energy.
\end{abstract}


\maketitle

\section{Introduction}

Vacuum energy is a concept arising in quantum field theory and was
first shown by Casimir \cite{Casimir:OnTheAttraction} to have an
observable effect on two perfectly conducting parallel plates, causing
them to attract.  Since then, experiments with various physical
geometries have confirmed the effects of vacuum energy (see
\cite{Boyer:QuantumZero-Point, Plunien:TheCasimirEffect, p:BMM:NDCE,
  Milton:TheCasimirEffect}).

In time-independent situations the vacuum energy is formally given by
\begin{equation}
  E = \frac12 \sum_n k_n
  \label{eq:quickref}
\end{equation}
where $k_n^2$ are the eigenvalues of a Hamiltonian, $H$, where throughout
we take $\hbar=1=c$.  The above expression arises in quantum
field theory in the context of cavities and cosmological models
\cite{Plunien:TheCasimirEffect}, and it is formally divergent.  To get
meaningful result from this expression, the vacuum energies for two
different configurations are subtracted from one another
\cite{Boyer:QuantumZero-Point}.  To accomplish this in a systematic
way, we employ an ultra-violet cutoff defining the energy as the
regular part of
\begin{equation}\label{eq:vc_defn}
E(t) = \frac 1 2 \sum_{n} k_n e^{-k_n t},
\end{equation}
as $t\to0$.  To evaluate (\ref{eq:vc_defn}), it is sometimes
convenient to employ the trace of the cylinder kernel, $T(t) = \sum_n
e^{-k_n t}$, \cite{p:SAF:GLVECOT:}.  In this way, $E(t) = - T'(t)/2$.
The singular term in the expansion of $E(t)$ is related to the vacuum
energy density of free space, and physical justification for its
removal is described, for example, in \cite{Boyer:QuantumZero-Point}
(for systems similar to those considered here, see
\cite{Cavalcanti:CasimirForce,Me:PistonPaper}).

A widely employed method of calculation of the vacuum energy is
expanding it into a sum over classical paths \cite{Brown:VacuumStress,
  p:JR:CFPTM, Schaden:Infinity-free, p:SAF:POSOSVE,
  Jaffe:CasimirEffect, Liu:RobinPlate}.  The expansion is usually done
by the method of images, or ``multiple reflections'', leading to a sum
over all closed paths.  It has been argued in
\cite{Schaden:Infinity-free} that for certain geometries restricting
the sum to include only the periodic paths (``semiclassical
evaluation'') correctly reproduces asymptotic behavior of the vacuum
energy and is much simpler to evaluate.  A periodic path comes back to
the starting point with the same momentum, while a closed path might
not.  Another popular approximation predicts the sign of the
vacuum energy by considering only short orbits
\cite{Sch06,Me:PistonPaper}.  This implicitly assumes that the
convergence of the complete series is sufficiently fast.

In the present paper we aim to contribute to this discussion by
studying the vacuum energy on quantum graphs (for another model where
similar questions are addressed, see \cite{LebMonBoh01}) .  Quantum
graphs are often used as mathematical models that exhibit the relevant
phenomena while being sufficiently simple to allow mathematical
treatment.  We compare the method of images with the direct
application of the trace formula (which is exact on graphs) and
demonstrate that the outcome is the same.  This is done by showing
that the contribution of the ``bounce paths'' --- the paths that are
closed but not periodic --- is identically zero.  We also prove that
the resulting sums converge, giving an estimate for the rate of
convergence, and can be differentiated term-by-term with respect to
the topological parameters present in the model.

One of the main reasons for the success of quantum graphs as models
(for example, of quantum chaos, see \cite{GS06} for a review) is the
existence of an exact trace formula on quantum graphs.  A trace
formula is a relation between the spectrum and the set of periodic
orbits of the system.  For graphs, the trace formula was first found
by Roth \cite{Roth:FirstTrace} and then by Kottos and Smilansky
\cite{Kottos:QuantumChaos}.  Subsequent studies of the mathematical
properties of the trace formula on graphs included works by Kostrykin,
Potthoff and Schrader \cite{Kostrykin:HeatKernels} and Winn
\cite{Winn:Trace}.

The trace formula of \cite{Kottos:QuantumChaos} gives an expression
for the density of states $d(k)$ defined as
\begin{equation}
  \label{eq:density_def}
  d(k) = \sum_{n=1}^\infty \delta(k-k_n),
\end{equation}
where $\delta(\cdot)$ is the Dirac delta-function.  The vacuum energy
is then, formally,
\begin{equation}
  \label{eq:ve_via_dens}
  E(t) = \frac12 \int k e^{-k t} d(k) dk,
\end{equation}
quickly leading to the result (for $k$-independent scattering
matrices, see Sec.~\ref{sec:graphs}),
\begin{equation}
  \label{eq:answer}
  E_c = - \frac 1 {2 \pi} \sum_{n=1}^\infty
  \sum_{p \in \PP} \frac{A_p}{\ell_p n_p}.
\end{equation}
The sum is over the set of all periodic paths $\PP$ on the graph, $\ell_p$ is the metric
length of the path $p$, $n_p$ is the period (the number of
bonds) of the path and $A_p$ is its stability amplitude (see
Section~\ref{sec:formal} for definitions).

After introducing the notation in Section~\ref{sec:graphs} and
considering some explicit examples in Section~\ref{sec:examples}, we
prove the mathematical correctness of the calculation outlined above.
This is done in Section~\ref{sec:conv}, where the convergence of
(\ref{eq:answer}) is also analyzed.  In Section~\ref{sec:deriv} we
show that we can differentiate (\ref{eq:answer}) with respect to
individual bond lengths, showing the smoothness ($C^\infty$) of $E_c$
as a function of lengths.

In Section~\ref{sec:images} we briefly discuss the method of images
(covered more fully in \cite{t:W:VEQG}) and show that the
contributions from the bounce paths cancel.  Finally, in
Section~\ref{sec:averages} we discuss random matrix models of the
vacuum energy.  As expected, in such models the average vacuum energy
is zero and there is no preferred sign to the Casimir force.  However,
by analyzing the second moment of the energy we confirm an earlier
observation by Fulling \cite{Fulling:Private,p:SAF:GLVECOT:} that
level repulsion tends to decrease the {\em magnitude\/} of the energy.

\section{Vacuum energy and quantum graphs}\label{sec:graphs}

Quantum graphs were introduced as a model of vacuum energy by Fulling
\cite{Fulling:Snowbird} who considered the effect of the energy
density near a quantum graph vertex by constructing the cylinder
kernel for an infinite star graph (a graph with one vertex and $B$
bonds extending to infinity).  The quantum field theory origins of
this in a graph context were given by Bellazzini and Mintchev
\cite{Bellazzini:QFonSGraphs}. The vacuum energy expression for
quantum graphs obtained in the present manuscript was also used in
\cite{Me:PistonPaper} where the convergence was investigated
numerically (we prove rigorous estimates here).

We start by briefly recalling the terminology of the quantum graph
model, see \cite{Kuchment:QGraphsI} for a general review of quantum
graphs.  We consider a finite metric graph $\Gamma$ consisting of
a set of vertices $\cV$, and a set of bonds $\cB$.
A (undirected) bond $b$ connecting the vertices $v$ and $w$ is denoted by
$\{v,w\}$.  Each bond $b$ is associated with a closed interval
$[0,L_b]$, thus fixing a preferred direction along the bond (from $0$
to $L_b$).  This direction can be chosen arbitrarily.  If the
direction from $v$ to $w$ is chosen, the bond $b=\{v,w\}$ gives rise
to two directed bonds, $b^+=(v,w)$ and $b^-=(w,v)$.  Whenever the
distinction between $b^+$ and $b^-$ is unimportant, we will denote the
directed bonds by Greek letters: $\alpha$, $\beta$.  In addition, the reversal
$\alpha$ is denoted by $\bar{\alpha}$ (e.g.\ if $\alpha = b^+$, then
$\bar{\alpha} = b^-$).
We will denote by $B$ the number of bonds $|\cB|$; correspondingly,
the number of directed bonds is $2B$.
The length of the
directed bond is naturally determined by the length $L_b$ of the underlying
undirected bond.  The total length of $\Gamma$ is $\cL=\sum_{b\in \cB} L_b$.
We will denote by $\L=\textrm{diag}\{L_1,\dots,L_B,L_1,\dots,L_B \}$
the diagonal $2B\times 2B$ matrix of directed bond lengths.

In this article we study the spectrum of the negative Laplacian
on the graph.  The Laplacian acts on the Hilbert space ${\mathcal
  H}(\Gamma) := \bigoplus_{b\in \cB} H^2 \bigl( [0,L_b] \bigr)$ of
(Sobolev) functions defined on the bonds of the graph.  On the bond
$b$ it acts as the 1-dimensional differential operator $-\frac{\ud^2
}{\ud x_b^2}$.  A domain on which the Laplacian is self-adjoint may be
defined by specifying matching conditions at the vertices of $\Gamma$,
see e.g.  \cite{ Kostrykin:Kirchhoff-I, Kostrykin:Kirchhoff-II,
  Kostrykin:TheGeneralized, Kuchment:QGraphsI}.

To specify the matching conditions, let $f$ be a function in
${\mathcal H}(\Gamma)$.  For a vertex $v$ of degree $d$ we denote by
$\bof^{(v)}$ the vector of values of $f$ at $v$,
$\bof^{(v)}=(f_{b_1}(v),\dots,f_{b_{d}}(v))^T$, where $f_b(v)=f_b(0)$
if $b=\{v,w\}$ is oriented from $v$ to $w$ and $f_b(v)=f_b(L_b)$
otherwise.  Furthermore, let $\bog^{(v)}$ denote the vector of
outgoing derivatives of $f$ at $v$,
$\bog^{(v)}=(f'_{b_1}(v),\dots,f'_{b_{d}}(v))^T$, i.e.
$f'_b(v)=f'_b(0)$ if $b=\{v,w\}$ is oriented from $v$ to $w$ and
$f'_b(v)=-f'_b(L_b)$ otherwise.  Matching conditions at $v$ can be
specified by a pair of matrices $\A^{(v)}$ and $\B^{(v)}$ through the
linear equation
\begin{equation}\label{eq:matching}
\A^{(v)} \bof^{(v)} +\B^{(v)} \bog^{(v)}=\bz \ .
\end{equation}
The matching conditions define a self-adjoint operator if
$(\A^{(i)}, \B^{(i)})$ has maximal rank and $\A^{(i)}\B^{(i)\dagger}$ is
self-adjoint at each vertex (where a $\B^{(i)\dagger}$ represents the adjoint of $\B$).

A solution to the eigenvalue equation on the bond $b$,
\begin{equation}
  \label{eq:eveq}
  -\frac{\ud^2}{\ud x_b^2} \psi_b(x_b)=k^2 \psi_b(x_b),
\end{equation}
can be written as a linear combination of plane waves,
\begin{equation}
  \label{eq:plane-waves}
  \psi_b(x_b) = c_b \ue^{\ui k x_b} + \ch_b \ue^{-\ui k x_b} \ .
\end{equation}
where $c$ is the coefficient of an outgoing plane wave at $0$ and
$\ch$ the coefficient of the incoming plane wave at $0$.  A solution
on the whole graph can be defined by specifying the corresponding
vector of coefficients
$\mathbf{c}=(c_1,\dots,c_B,\ch_1,\dots,\ch_B)^T$.

The matching conditions at the vertex $v$ define a vertex scattering matrix
\begin{equation}
  \label{eq:S-matrix_defn}
  \vS^{(v)}(k)=-(\A^{(v)}+\ui k\B^{(v)})^{-1}(\A{(v)}-\ui k\B^{(v)}),
\end{equation}
see \cite{Kostrykin:Kirchhoff-I}.  $\vS^{(v)}$ is unitary and the
elements of $\vS^{(v)}$ are complex transition amplitudes which in
general depend on $k$.  However, for a large class of matching
conditions including the so called Kirchhoff or natural conditions the
S-matrix is independent of $k$.
Kirchhoff matching conditions require that
$\psi$ is continuous at the vertex and the outgoing
derivatives of $\psi$ at the vertex sum to zero.
These conditions may be
written in the form (\ref{eq:matching}) with matrices
\begin{equation}\label{eq:neumann_matrices}
\A=\left( \begin{array}{ccccc}
1& -1 & 0 & 0 & \dots \\
0 & 1 & -1 & 0 & \dots \\
& & \ddots & \ddots &  \\
0& \dots & 0 &  1 & -1 \\
0 &\dots & 0 & 0& 0 \\
\end{array} \right)  \qquad \B=\left( \begin{array}{cccc}
0& 0 & \dots & 0 \\
\vdots & \vdots &  & \vdots \\
0& 0 & \dots & 0 \\
1 &1 & \dots & 1 \\
\end{array} \right) \ .
\end{equation}
Substituting in (\ref{eq:S-matrix_defn}) leads to $k$-independent transition
amplitudes
\begin{equation}
  \label{eq:neumann_amplitudes}
  [\vS]_{ij}=\frac{2}{d}-\delta_{ij} \ ,
\end{equation}
where $d$ is the degree of $v$.

The matrix $\vS^{(v)}$ relates incoming and outgoing plane wave coefficients
at $v$, $\mathbf{c}^{(v)}=\vS \hat{\mathbf{c}}^{(v)}$.
Collecting together
transition amplitudes from all the vertices of a graph we may define the
familiar $2B\times 2B$ \emph{bond scattering matrix} $\bS$ \cite{KS99},
\begin{equation}
  \label{eq:bondS-matrix_defn}
  [\bS]_{(v',w') (v,w)} = \delta_{w,v'} [\vS^{(w)}]_{(v,w')} \ .
\end{equation}
We shall also need the \emph{quantum evolution operator} $\U = \bS \ue^{\ui k
  \L}$, which acts on the vector of $2B$ plane wave coefficients
indexed by directed bonds.  For a general graph, the spectrum can be
computed as the zeros of the equation
\begin{equation}
  \label{eq:det_eq}
  \det(\I -  \bS e^{ik\L}) = 0.
\end{equation}
This formula goes back at least to \cite{vB85}; for a discussion of
scattering matrices of different types we refer the reader to
\cite{KS99}.

The spectral theory of quantum graphs is often extended to included quantum evolution operators defined by specifying \emph{a priori} a set of unitary vertex scattering matrices as in \cite{p:SS:SSFQG,p:T:USMESS}.
The vertex scattering matrices are typically chosen to be $k$ independent and to have other desirable features, for instance transition amplitudes of equal magnitude as is the case if the scattering matrix is a fast Fourier transform matrix
\begin{equation}
  [\vS]_{ij} = \frac{1}{\sqrt{d}} \ue^{\ui \frac{2\pi}{d} ij}  \ .
\end{equation}
Such a scattering matrix would be difficult to produce from matching conditions of a self-adjoint operator at a vertex.
For such a quantum evolution operator the spectrum is still defined by (\ref{eq:det_eq}) however the scattering matrices in general no longer correspond to a self-adjoint realization of the Laplace operator on the graph.

In the following we work with $k$-independent scattering matrices which can either be considered to come from a self-adjoint Laplace operator or through specifying unitary vertex scattering matrices directly.

\section{Some explicit examples}\label{sec:examples}

\subsection{Star graph with bonds of equal length}

One case where the vacuum energy can be computed explicitly is the
quantum star graph with bonds of equal length.  This example was first
considered in \cite{Me:PistonPaper} and here, for completeness, we
summarize the computation.

\begin{figure}[htb]
  \begin{center}
    \includegraphics[width=4cm]{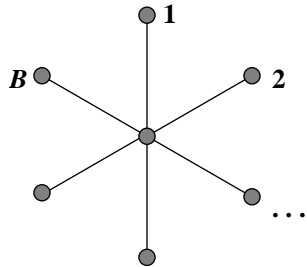}
    \caption{A star graph with $B$ bonds}\label{fig:star}
  \end{center}
\end{figure}

Consider a star graph (see Fig.~\ref{fig:star}) with $B$ bonds.  The
bond $b$ has length $L_b$.  We consider the  equation
(\ref{eq:eveq}) with Neumann conditions.  At the central vertex, this
translates into
\begin{equation}
  \label{eq:cond_central}
  \sum_b \psi'_b(0) = 0, \qquad \psi_1(0) = \ldots = \psi_B(0)
\end{equation}
and, at the end-vertices, into
\begin{equation}
  \label{eq:cond_end}
  \psi_b'(L_b) = 0, \quad \forall b.
\end{equation}

Solutions of (\ref{eq:eveq}) together with (\ref{eq:cond_end}) can be
written as
\begin{displaymath}
  \psi_b(x) = C_b \cos(k(L_b-x)).
\end{displaymath}
Imposing (\ref{eq:cond_central}) we conclude that the spectrum consists of the
solutions to
\begin{displaymath}
  Z(k) = \sum_{b=1}^B \tan(kL_b) = 0,
\end{displaymath}
if the lengths are rationally independent.  We note that since $Z(k)$ is an
increasing function, there is exactly one zero of $Z(k)$ between each pair of
consecutive poles.  Because of rational independence, the poles of different
tangents do not coincide.  If we lift the restriction on lengths, in addition
to the zeros of $Z(k)$ we have the following eigenvalues: $\kappa$ is an
eigenvalue of multiplicity $m$ if there are $m+1$ lengths $L_{b_j}$ such that
$\kappa$ is a pole of each $\tan(k L_{b_j})$.

In particular, if all lengths are equal, $L_1 = \ldots = L_B = L$, $Z(k)$
is simply $B\tan(kL)$.  Each zero of $Z(k)$ is a simple eigenvalue, while
each pole is an eigenvalue of multiplicity $B-1$.  We can now evaluate
\begin{eqnarray*}
  \fl  T(t) = \sum_{n=1}^\infty e^{-tk_n}
    = \sum_{n=0}^\infty e^{-n\pi t / L}
    + (B-1) \sum_{n=0}^\infty e^{-(2n+1)\pi t / 2L}
    = \frac{1 + (B-1) e^{-\pi t / 2L}}{1 - e^{-\pi t/L}}\\
    = \frac{BL}{\pi}t^{-1} + \frac12  - \frac{(B-3)\pi}{24L}t+O(t^2)
\end{eqnarray*}
Now we take the regular part and the limit $t\to0$ of $-T'(t)/2$ which
gives us the vacuum energy,
\begin{equation}
  \label{Ec_star_answer}
  E_c = \frac{(B-3)\pi}{48 L}.
\end{equation}

\subsection{General graphs with bonds of equal length}

If all bond lengths of the graph are equal to $L$, we can use equation
(\ref{eq:det_eq}) to explicitly describe the infinite spectrum of the
graph in terms of the finite spectrum of $\bS$.  Indeed let $\LL$ be
the diagonal matrix of the eigenvalues of $\bS$.  Then $e^{ik\L} =
e^{ikL}\I$ and, therefore, $\det(\I - \bS e^{ikL\I}) = \det(\I - \LL
e^{ikL})$.  Consequently the solutions of equation~(\ref{eq:det_eq})
are the values $k$ such that
\begin{equation}
  \label{eq:k_solns}
  e^{ikL} e^{i\theta_j} = 1
\end{equation}
for some $j$.  Here by $e^{i\theta_j}$ we denoted the $j$-th eigenvalue of
(unitary) matrix $\bS$.  Thus the $k$-spectrum is
\begin{equation}
  \label{eq:k_spec}
  \bigcup_{j=1}^{2B} \left\{ \frac{2\pi n - \theta_j}{L} \right\}_{n=1}^\infty
\end{equation}
where we choose $\theta_j$ to lie between $0$ and $2\pi$.

Now we compute the trace of the cylinder kernel $T(t)$
\begin{eqnarray}
    \label{eq:T_comp_gen_graphs}
    T(t) &= \sum_{n=1}^\infty e^{-tk_n}
    = \sum_{j=1}^{2B} e^{t\theta_j/L} \sum_{n=1}^\infty e^{-2\pi nt/L}
    = \left(e^{2\pi t/L} - 1\right)^{-1} \sum_{j=1}^{2B} e^{t\theta_j/L}\\
    &= \sum_{j=1}^{2B} \left[\frac{L}{2\pi t} + \frac{\theta_j-\pi}{2\pi}
      + \frac{3 \theta_j^2  + 2\pi^2  - 6 \theta_j\pi}{12 L \pi}t
      + O(t^2) \right].
\end{eqnarray}
Thus, the vacuum energy is
\begin{equation}
  \label{eq:Ec_gen_graph_equal}
  E_c = -\sum_{j=1}^{2B} \frac{3 \theta_j^2  + 2\pi^2  - 6 \theta_j\pi}{12 L
  \pi} = -\frac{\pi}{L} \sum_{j=1}^{2B} B_2(\theta_j/2\pi),
\end{equation}
where $B_2(\cdot)$ is the second Bernoulli polynomial (compare to the $B=1$
case discussed in \cite{p:SAF:GLVECOT:}).

As an example of application of formula (\ref{eq:Ec_gen_graph_equal})
consider again the star graph with equal bond lengths.  The
eigenphases $\theta_j$ of the S-matrix of a star graph are $0$, $\pi$,
$\pi/2$ and $3\pi/2$, the latter two with multiplicity $B-1$.
Substituting into equation (\ref{eq:Ec_gen_graph_equal}) one can
recover (\ref{Ec_star_answer}).

\subsection{Graphs with bonds of rational length}

By introducing the Neumann vertices of degree 2 we do not change the spectrum
of the graph and therefore the vacuum energy.  On the other hand, if the bonds
of the graph are rational (up to an overall factor), by introducing such
``dummy'' vertices we can convert the original graph into a graph with bonds
of equal length.  The number of bonds (and the dimension of the scattering
matrix $\bS$) will increase as a result, but the vacuum energy will still be
explicitly computable using equation~(\ref{eq:Ec_gen_graph_equal}).

Moreover, one can conceivably approximate rationally independent
lengths by rational ones and use the result as a numerical
approximation to the true vacuum energy.  For this approach to work
one needs to know, \textit{a priori}, that $E_c$ is continuous as a
function of bond lengths.  This question is one of the main subjects
of Section~\ref{sec:trace}.

\section{Vacuum energy via the trace formula}
\label{sec:trace}

\subsection{Formal calculation}
\label{sec:formal}

In this section we perform a formal calculation of the vacuum energy
$E_c$ using the trace formula.  We shall investigate the rigor of the
manipulations in Section~\ref{sec:conv}.

The trace formula (see, e.g. \cite{GS06}) connects the
spectrum $\{k_n\}$ of the graph with the set of all periodic
orbits (or periodic paths) on the graphs.  A periodic path of period $n$
is a sequence $(\alpha_1,\alpha_2,\ldots, \alpha_n)$ of directed bonds which
satisfy $[\bS]_{\alpha_{j+1},\alpha_j} \neq 0$ for all
$j=1,\ldots,n$ (the index $j+1$ is taken modulo $n$).  A periodic
orbit is an equivalence class of periodic paths with respect to the
cyclic shift $(\alpha_1,\alpha_2,\ldots, \alpha_n) \mapsto (\alpha_2,\ldots,
\alpha_n, \alpha_1)$.  We denote by $\PP_n$ the set of all periodic
paths of period $n$ and by $\PP$ the set of periodic paths of all
periods.  The the trace formula can be written as
\begin{equation}
  \label{eq:trace_formula}
  d(k) \equiv \sum_{n=1}^\infty \delta(k-k_n) = \frac{\cL}\pi
  + \frac1\pi \re \sum_{p\in\PP} A_p \frac{\ell_p}{n_p}
  e^{\ui k\ell_p},
\end{equation}
where $\cL$ is the total length of the graph (the sum of the bond
lengths), $n_p$ is the period of the periodic path $p$, $\ell_p =
\sum_{j=1}^{n_p} L_{\alpha_j}$ is the length of $p$ and $A_p =
\prod_{j=1}^n [\bS]_{\alpha_{j+1},\alpha_j}$ is its amplitude.  Using
the trace formula and equation~(\ref{eq:vc_defn}) we can formally
compute the vacuum energy.  Indeed,
\begin{eqnarray}
  \sum_{n=1}^\infty k_n e^{-tk_n} &= \int_0^\infty ke^{-kt} d(k) dk \\
  &= \frac{\cL}{2\pi} \int_0^\infty ke^{-kt} dk
  + \frac1\pi\re \sum_{p\in\PP} A_p \frac{\ell_p}{n_p}
  \int_0^\infty ke^{-kt+\ui k\ell_p} dk \\
  &= \frac{\cL}\pi t^{-2}
  + \frac1\pi \re \sum_{p\in\PP} \frac{A_p \ell_p}{n_p(t-\ui \ell_p)^2}
\end{eqnarray}
Removing the divergent Weyl term $\cL/\pi t^2$ due to regularization and taking the limit $t\to0$
leads to the following simple expression for the vacuum energy,
\begin{equation}
  \label{eq:vc_trace_expansion}
  E_c = -\frac1{2\pi} \re \sum_{p\in\PP} \frac{A_p}{\ell_p n_p}.
\end{equation}

\subsection{Equal bond lengths; equivalence to (\ref{eq:Ec_gen_graph_equal})}
\label{sec:equal}

In the case of equal bond lengths (\ref{eq:vc_trace_expansion}) should
be equivalent to the sum of second Bernoulli polynomials
(\ref{eq:Ec_gen_graph_equal}).  If the length of each bond is $L$ an
orbit that visits $n$ bonds has length $nL$ and we may rewrite
(\ref{eq:vc_trace_expansion}) as a sum over the topological length
$n$ followed by a sum over the set of all periodic paths visiting $n$
bonds, $\PP_n$,
\begin{eqnarray}
  E_c &= -\frac{1}{2\pi L} \re \sum_{n=1}^{\infty} \frac{1}{n^2}
  \sum_{p\in\PP_n} A_p \\
  &= -\frac{1}{2\pi L} \sum_{n=1}^{\infty} \frac{1}{n^2}
  \sum_{\al_1=1}^{2B} \dots \sum_{\al_n=1}^{2B} \re
  (S_{\al_1\al_2}S_{\al_2\al_3}\dots S_{\al_n\al_1}) \\
  \label{eq:Ec_power S}
  &= -\frac{1}{2\pi L} \sum_{n=1}^{\infty}
  \frac{1}{2n^2} \big( \tr \bS^n + \tr(\bS^\dagger)^n \big) \\
  &= -\frac{1}{2\pi L} \sum_{n=1}^{\infty} \sum_{j=1}^{2B} \frac{\cos
    n \theta_j}{n^2},
\end{eqnarray}
where $e^{i\theta_j}$ are the eigenvalues of the matrix $\bS$ with
$0\leqslant \theta \leqslant 2\pi$.  The sum over $n$ can be expressed
in a closed form, see Abramowitz and Stegun \cite{B:AS:HoMF}, formulae 27.8.6,
\begin{equation}\label{eq:sum Abr Ste}
  \sum_{n=1}^{\infty} \frac{\cos (n\theta)}{n^2}=
  \frac{3\theta^2 + 2\pi^2 -6\pi \theta}{12} = \pi^2 B_2(\theta/2\pi),
\end{equation}
where $B_2(\cdot)$ is the second Bernoulli polynomial.  Consequently we
recover expression (\ref{eq:Ec_gen_graph_equal}).

\subsection{Convergence of (\ref{eq:vc_trace_expansion}); vacuum energy as a
  function of the bond lengths}
\label{sec:conv}

We shall now present a rigorous derivation of
equation~(\ref{eq:vc_trace_expansion}).

\begin{theorem}
  The vacuum energy of the graph, defined as
  \begin{displaymath}
    E_c = \frac12 \lim_{t\to0+} \left[\sum_{n=1}^\infty k_n e^{-tk_n} -
    \frac{\cL}{\pi t^2} \right],
  \end{displaymath}
  is given by
  \begin{eqnarray}
    \label{eq:ve_sum_int_tr}
    E_c &= \frac{1}{2\pi} \sum_{n=1}^\infty
    \frac{1}{n} \re \int_0^\infty \tr\left(\bS e^{-s\L}\right)^n \ud s \\
    \label{eq:ve_sum_po}
    &= -\frac1{2\pi} \re \sum_{n=1}^\infty \sum_{p\in\PP_n}
    \frac{A_p}{\ell_p n_p} ,
  \end{eqnarray}
  where $\PP_n$ denotes the set of all periodic paths of period $n$.
  The vacuum energy is smooth ($C^\infty$) as a function of bond
  lengths on the set $\{L_b>0\}$.
\end{theorem}

\begin{remark}
  The sum over the periodic orbits in (\ref{eq:ve_sum_po}) is finite
  for each $n$.  We will show, in particular, that the sum over $n$ is
  absolutely and uniformly convergent.  More precisely we will derive
  the following bound,
  \begin{equation}
    \label{eq:jg_estimate}
    \left| \sum_{p\in{\cal P}_n} \frac{A_p}{\ell_p n} \right|
    \leq \frac{2B}{n^2 L_\mathrm{min}},
  \end{equation}
  where $2B$ is the number of (directed) bonds and $L_\mathrm{min}$ is the
  minimal bond length.  This estimate shows that, if the (finite!) sum over
  periodic orbits of a fixed length is performed first, the series in
  (\ref{eq:ve_sum_po}) becomes absolutely convergent.  Moreover, it
  is uniformly convergent with respect to the change in bond length as long as
  $L_\mathrm{min}$ remains bounded away from zero.

  We would like to mention that our estimate (\ref{eq:jg_estimate})
  agrees with the numerical results of \cite{Me:PistonPaper}, even
  though in \cite{Me:PistonPaper} the ordering of the periodic orbits
  was different (according to the metric length $\ell_p$ rather than
  topological length $n$).
\end{remark}

\begin{proof}[Proof of (\ref{eq:ve_sum_int_tr})-(\ref{eq:ve_sum_po}).]

  The $C^\infty$ part of the proof will be given in the following section.

  We start with the definition of the vacuum energy and integrate by
  parts,
  \begin{equation}
    \label{eq:use_IDS}
    \sum_{n=1}^\infty k_n e^{-tk_n} = \int_0^\infty (tk-1) e^{-tk} N(k) dk,
  \end{equation}
  where $N(k)$ is the {\em integrated density of states\/} (IDS), a piecewise
  constant, increasing function
  \begin{equation}
    \label{eq:IDS}
    N(k) = \#\{n : 0<k_n<k\}.
  \end{equation}
  The integrated density of states $N(k)$ can be split into two parts,
  \begin{equation}
    \label{eq:IDS_split}
    N(k) = \mbox{const} + \frac{k\cL}\pi + N^{\mathrm{osc}}(k).
  \end{equation}
  The first two terms are unimportant: The first term makes no contribution
  in the integral, and the second term is removed at the regularization stage.
  The oscillatory part possesses an
  expansion, see \cite{GS06}, equation (5.24),
  \begin{equation}
    \label{eq:trace_N_imag}
    N^{\mathrm{osc}}(k+\ui\eps) = \frac{1}{\pi} \im \sum_{n=1}^\infty
    \frac{1}{n} \tr \U^n(k+\ui \eps),
  \end{equation}
  where $\U = \bS e^{\ui k \L}$.

  This expansion is absolutely convergent as long as $\eps>0$ since
  the matrix $\U^n(k+\ui \eps)$ is then subunitary (all eigenvalues
  lie within a circle of radius strictly less than 1).  As $\eps\to0$,
  $N^{\mathrm{osc}}(k+\ui\eps)$ converges to $N^{\mathrm{osc}}(k)$
  pointwise almost everywhere.  Moreover,
  $\left|N^{\mathrm{osc}}(k+\ui\eps)\right|$ is uniformly bounded by
  the number of bonds $B$ (in other systems one can show that the Weyl
  law implies that $|N^{\mathrm{osc}}(k)| = O(k^d)$ as $k \rightarrow \infty$
  where $d$ is the dimension of the system).
  Therefore,
  \begin{equation}
    \label{eq:use_IDS_eps}
    E_c = -\frac12 \lim_{t\to0} \lim_{\eps\to0} \int_0^\infty (tk-1) e^{-tk}
    N^{\mathrm{osc}}(k+\ui\eps) \ud k,
  \end{equation}
  and, using the convergence of expansion (\ref{eq:trace_N_imag}),
  \begin{equation}
    \label{eq:use_IDS_fin}
    E_c = -\frac{1}{2\pi} \lim_{t\to0} \lim_{\eps\to0}
    \sum_{n=1}^\infty
    \frac{1}{n} \im \int_0^\infty (tk-1) e^{-tk} \left[ \tr \U^n(k+\ui \eps)
    \right] \ud k.
  \end{equation}

  We will now show that the integral
  \begin{equation}
    \label{eq:integral}
    R_n = \int_0^\infty (tk-1) e^{-tk} \left[\tr \U^n(k+\ui \eps) \right] \ud k
  \end{equation}
  is absolutely bounded by $1/n$.  Thus the series is absolutely convergent
  uniformly in $\eps$ and $t$ and we can take the limits inside the sum.

  A typical term in the (finite!) expansion of the trace is $A_p
  e^{\ui k \ell_p} e^{-\eps \ell_p}$.  The two exponential factors,
  $e^{-tk}$ and $e^{\ui k \ell_p}$, ensure that the integrand is
  exponentially decaying in $k$ in the first quadrant of $\mathbb{C}$.
  Therefore we can rotate the contour of integration to the imaginary
  line, $k=\ui s$.  The integral becomes
  \begin{equation}
    \label{eq:int_rotated}
    R_n = \ui \int_0^\infty (\ui st - 1) e^{-\ui st}
    \tr\left(\bS e^{-(s+\eps)\L}\right)^n ds.
  \end{equation}
  We estimate
  \begin{equation}
    \label{eq:trace_ineq}
    \left|\tr\left(\bS e^{-(s+\eps)\L}\right)^n\right|
    \leq \sum_{j=1}^{2B} |\lambda_j(s+\epsilon)|^n,
  \end{equation}
  where $2B$ is the size of the matrix $\bS$ and $\lambda_j$ is $j$-th
  eigenvalue of the matrix $\bS e^{-(s+\eps)\L}$.  According to a
  familiar argument (see, e.g. \cite{Wey49}), the maximal
  $|\lambda_j|$ is bounded from above by the maximal singular value of
  the matrix.  The singular values are square roots of the eigenvalues
  of
  \begin{displaymath}
    \left(\bS e^{-(s+\epsilon)\L}\right)^\dagger
    \left(\bS e^{-(s+\epsilon)\L}\right) =  e^{-2(s+\epsilon)\L},
  \end{displaymath}
  which is a diagonal matrix with the maximal entry
  $e^{-2(s+\epsilon)L_{\min}}$ ($L_{\min}$ is the smallest bond length
  of the graph).  Thus we can estimate
  \begin{equation}
    \label{eq:trace_est}
    \left|\tr\left(\bS e^{-(s+\eps)\L}\right)^n\right|
    \leq e^{-nsL_{\min}} 2B.
  \end{equation}
  Finally,
  \begin{equation}
    \label{eq:int_est}
    |R_n| \leq 2B \int_0^\infty |\ui st-1| e^{-nsL_{\min}} ds \sim \frac1n.
  \end{equation}

  We notice that the integrand of (\ref{eq:int_rotated}) can be
  absolutely bounded by $e^{-ns(L_{\min}-\delta)}$ for arbitrarily
  small $\delta$ and sufficiently small $t$.  Thus, having brought the
  limits inside the sum, we can use the dominated convergence theorem
  to bring them inside the integral.  Taking the limit $t\to0$ and
  $\eps\to0$ inside integral (\ref{eq:int_rotated}) produces
  \begin{equation}
    \label{eq:uzhe_konec}
    E_c = \frac{1}{2\pi} \sum_{n=1}^\infty
    \frac{1}{n} \re \int_0^\infty \tr\left(\bS e^{-s\L}\right)^n \ud s.
  \end{equation}
  Now we expand the trace,
  \begin{equation}
    \label{eq:trace_expansion}
    \tr\left(\bS e^{-s\L}\right)^n = \sum_{p\in\PP_n} A_p e^{-s \ell_p},
  \end{equation}
  and integrate term by term to recover (\ref{eq:ve_sum_po}).
\end{proof}

\begin{remark}
  The basic idea of the proof, shifting the convergence into the
  subunitary matrix (see equation (\ref{eq:int_rotated}) and
  (\ref{eq:trace_est})) can also be used to study the convergence of
  the trace formula itself.  This was done in \cite{Win07}.
\end{remark}

\begin{remark}
  Since $\bS$ is $k$-independent, the trace in the integral is real
  and we do not need to take the real part in equations
  (\ref{eq:ve_sum_int_tr})-(\ref{eq:ve_sum_po}).  Indeed, the
  following lemma easily follows from the definition of $\bS$ and
  \cite[Prop 2.4]{Kostrykin:HeatKernels} applied to the matrices
  $\vS^{(v)}$.
  \begin{lemma}
    \label{lem:k_indep_S}
    The $\bS$-matrix
    of a graph is $k$-independent if and only if it satisfies
    \begin{equation}
      \label{eq:S_cond}
      \J\bS\J = \bS^\dagger,
    \end{equation}
    where $\J$ is defined by $J_{\alpha,\beta} = \delta_{\alpha,\bar{\beta}}$
    ($\bar{\beta}$ is the reversal of $\beta$ as defined in section \ref{sec:graphs}).
  \end{lemma}

  Now the complex conjugate of $\tr\left(\bS e^{-s\L}\right)^n$ is
  \begin{equation}
    \label{eq:cc_trace}
    \tr\left(e^{-s\L} \bS^\dagger\right)^n
    = \tr\left(e^{-s\L} \J\bS\J\right)^n
    = \tr\left(\bS\J e^{-s\L} \J\right)^n
    = \tr\left(\bS e^{-s\L}\right)^n,
  \end{equation}
  where we used the fact that the length of a bond is invariant with respect to
  direction reversal and, therefore, $\J e^{-s\L}\J = e^{-s\L}$.
\end{remark}

\subsection{Derivatives of the vacuum energy}\label{sec:deriv}

\begin{proof}[Proof of differentiability of $E_c$]
  We differentiate the expression (\ref{eq:ve_sum_int_tr}) term by term
 and show that the result is also absolutely convergent.  This can be
  done by using the following bound on the partial derivatives of
  $\tr \big( \bS \ue^{-s\L} \big)^n$,
  \begin{equation}
    \label{eq:partial_derivative_bound}
    \left| \frac{\partial^{m_1} \dots \partial^{m_B}}{
        \partial L_1^{m_1}\dots  \partial L_B^{m_B}}
      \tr \big( \bS \ue^{-s\L} \big)^n \right|
    \leqslant \frac{2B \ue^{-snL_{\min} /2}}{(L_{\min}/2)^{|\m|}}.
  \end{equation}
  Before proving (\ref{eq:partial_derivative_bound})
  we note that it implies the following bound,
  \begin{equation}
    \left| \frac{1}{n} \int_0^\infty
      \frac{\partial^{m_1} \dots \partial^{m_B}}{
        \partial L_1^{m_1}\dots  \partial L_B^{m_B}}
      \tr \big( \bS \ue^{-s\L} \big)^n \ud s \right|
    \leqslant \frac{2B}{n^2(L_{\min}/2)^{|\m|+1}} \ .
  \end{equation}
  Consequently
  \begin{equation*}
    \sum_{n=1}^\infty \frac{1}{n}
    \int_0^\infty \frac{\partial^{m_1} \dots \partial^{m_B}}{
      \partial L_1^{m_1}\dots  \partial L_B^{m_B}}
    \tr \big( \bS \ue^{-s\L} \big)^n \ud s
  \end{equation*}
  converges absolutely and we are therefore allowed to differentiate
  (\ref{eq:ve_sum_int_tr}) term by term.  We conclude that the vacuum
  energy is $C^\infty$ as a function of bond lengths.

  To prove bound (\ref{eq:partial_derivative_bound}) we use the Cauchy
  integral formula (see, e.g., \cite{b:MW:MMoP}),
  \begin{equation}\label{eq:cauchy}
    \fl \frac{\partial^{m_1} \dots \partial^{m_B}}{
      \partial L_1^{m_1}\dots  \partial L_B^{m_B}}
    \tr \big( \bS \ue^{-s\L} \big)^n =
    \frac{1}{(2\pi)^B} \int_0^{2\pi} \dots \int_0^{2\pi}
    \frac{\tr \big( \bS \ue^{-s(\L +\R(\gf))} \big)^n}{
      R_1^{m_1}\dots R_B^{m_B}}
    \ud \phi_1 \dots \ud \phi_B
  \end{equation}
  where $R_j=r_j\ue^{\ui \phi_j}$ and $\R (\gf) = \textrm{diag} \{
  R_1,\dots,R_B,R_1,\dots,R_B \}$.  Let $\vA=\bS \ue^{-s(\L +\R)}$.
  Since $\bS$ is unitary, we find $\vA^{\dagger}\vA=
  \ue^{-s(2\L+\R+\R^\dagger) }$.   The eigenvalues of $\vA$ are
  bounded by the maximal singular value of $\vA$,
  \begin{equation}
    \label{eq:sig val A}
    | \textrm{eig} ( \vA ) | \leqslant \max_{b=1\dots B}
    \, \ue^{-s (L_b +r_b\cos(\phi_b) )} \ .
  \end{equation}
  By choosing the radius $r_b=L_{\min} /2$ we see that
  \begin{equation}
    \label{eq:eig A bound}
    | \textrm{eig}( \vA ) | \leqslant
    \max_{b=1\dots B} \ue^{-s(L_b-L_{\min} / 2)} = \ue^{-sL_{\min}/2} \ ,
  \end{equation}
  and
  \begin{equation}
    \label{eq:trace_A_bound}
    \left|\tr \big( \bS \ue^{-s(\L +\R(\gf))} \big)^n\right|
    = \left|\tr \vA^n \right|
    \leq \sum_{j=1}^{2B} |\textrm{eig}(\vA)_j|^n
    \leq 2B \ue^{-snL_{\min}/2}.
  \end{equation}
  Using this bound in (\ref{eq:cauchy}) gives
  \begin{equation}\label{eq:bound in cauchy}
    \fl  \left| \frac{\partial^{m_1} \dots \partial^{m_B}}{
        \partial L_1^{m_1}\dots  \partial L_B^{m_B}}
      \tr \big( \bS \ue^{-s\L} \big)^n \right|
    \leqslant \frac{1}{(2\pi)^B} \int_0^{2\pi} \dots \int_0^{2\pi}
    \frac{2B \ue^{-snL_{\min}/2}}{(L_{\min}/2)^{|\m|}} \, \ud \phi_1 \dots
    \ud \phi_B
  \end{equation}
  which establishes (\ref{eq:partial_derivative_bound}).
\end{proof}

\section{Method of images expansion and the equivalence of two
  expansions}
\label{sec:images}

We can evaluate the trace of the cylinder kernel $T(t)$ by
constructing the kernel itself.  The cylinder kernel,
$T_{b b'}(t; x, y)$ where $x$ is measured on bond $b$ and $y$ on bond $b'$
satisfies the following equation on each bond $b$
\begin{equation}
  -\frac{\partial^2}{\partial x^2} T_{bb'}(t ; x , y)
  = \frac{\partial^2}{\partial t^2} T_{bb'}(t ; x , y)
  \label{eq:cylindereqn}
\end{equation}
for $t>0$, $T_{bb'} \rightarrow 0$ as $t \rightarrow \infty$, with
boundary conditions (\ref{eq:matching}) with respect
to the $x$ coordinate, and the initial condition $T_{bb'}(0; x,
y) = \delta_{bb'} \delta(x - y)$.  By separating variables in
(\ref{eq:cylindereqn}) it can be shown that the trace of $T$ is
\begin{equation}
  T(t) \stackrel{\mathrm{def}}{=}
  \sum_{b=1}^B \int_0^{L_b} T_{bb'}(t;x,x) \; \ud x
  = \sum_{n} e^{-k_n t} \ .
\end{equation}

The corresponding free space kernel (i.e. the solution of
equation~(\ref{eq:cylindereqn}) on the whole real line with the
initial condition $T_0(0;x-y) = \delta(x-y)$) is
\begin{equation}
  T_0(t;x-y) = \frac{t/ \pi}{t^2 + (x-y)^2}.
\end{equation}
We can apply the method of images (i.e.\ multiple reflection) to the
free space kernel to obtain the kernel on the graph.  The cylinder
kernel is then written in terms of paths which begin at $y$ on
the bond $b'$ and end at $x$ on the bond $b$ (the details of the
construction are given in \cite{t:W:VEQG}),
\begin{eqnarray}
  \fl T_{bb'}(t; x, y)
  = \delta_{bb'} T_0(t;x-y) \nonumber \\
  + \sum_{n=0}^\infty \sum_{\path \in \Paths_n} \Big[
  A_{b^+ \path\, b'^+} T_0(t; L_b + \ell_{\path} + x - y)
  + A_{b^+ \path\, b'^-} T_0(t ; -\ell_{\path} - x - y) \\\label{eq:moi}
  + A_{b^- \path\, b'^-} T_0(t; -\ell_{\path} - L_b + x -y)
  + A_{b^- \path\, b'^+} T_0(t; L_{b'} + \ell_{\path} + L_b - x - y)\Big] \nonumber
\end{eqnarray}
In the above expression, we denote the two directed bonds associated
with the undirected bond $b$ by $b^+$ and $b^-$.  The path of
topological length $n$, $\path = (\alpha_1, \ldots, \alpha_n)$ is an
$n$ vector of directed bonds and $\Paths_n$ is the set of all such
paths.  The metric length of a path is $\ell_{\path} = \sum_{j=1}^n
L_{\alpha_j}$ and $A_{\path} = [\bS]_{\alpha_n \alpha_{n-1}} \cdots
[\bS]_{\alpha_3 \alpha_2} [\bS]_{\alpha_2 \alpha_1}$ is the stability
amplitude of the path.  Note that even at the stage represented in
(\ref{eq:moi}) we have assumed the matrix $\bS$ is $k$-independent.

To obtain the trace, we let $b=b'$ and $y = x$.
While this corresponds to ``closing'' the paths, we do not always get
periodic paths in the topological sense of Section~\ref{sec:formal}.
Indeed, a periodic path would return to the initial point $x_b$ with
the same momentum (or direction), whereas when the paths corresponding to
the second and fourth terms of~(\ref{eq:moi}) return to $x_b$, the
momentum has the opposite sign, see Fig. \ref{fig:periodic_and_bounce}.
The latter paths we shall call {\em bounce paths}. The difficulty in
the method of images is in the handling of the bounce paths.

\begin{figure}[htb]
  \begin{center}
    \includegraphics[scale=0.4]{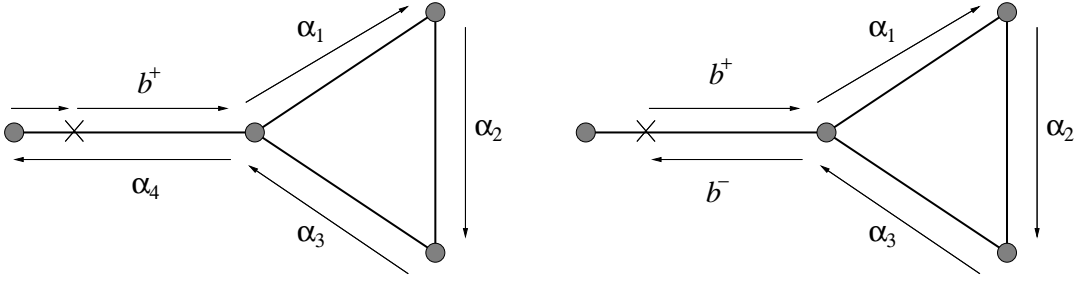}
    \caption{An example of a periodic path (left) and a bounce
      path (right).}\label{fig:periodic_and_bounce}
  \end{center}
\end{figure}

After integrating $T_{bb}(t; x, x)$ we break the formula into three
parts
\begin{equation}
T(t) = T_{\textrm{FS}}(t) + T_{\textrm{PO}}(t) + T_{\textrm{BP}}(t),
\end{equation}
where FS stands for free space, BP for bounce paths, and PO for
periodic orbits.  The three parts are (taking into account that $T_0(t;x)$
is even in $x$):
\begin{eqnarray}
  \label{eq:TFS}
  \fl T_{\textrm{FS}}(t) &= \sum_{b = 1}^B \int_0^{L_b} T_0(t;0) \; \ud x
  = T_0(t;0) \cL,\\
  \label{eq:TPO}
  \fl T_{\textrm{PO}}(t) &= \sum_{b = 1}^B \int_0^{L_b} \sum_{n=0}^\infty
  \sum_{\path \in \Paths_n}
  \Big[ A_{b^+ \path\, b^+} T_0(t;L_b + \ell_{\path})
  + A_{b^- \path\, b^-} T_0(t;\ell_{\path} + L_b) \Big] \; \ud x,\\
  \label{eq:TBP}
  \fl T_{\textrm{BP}}(t) &= \sum_{b = 1}^B \int_0^{L_b}
  \sum_{n=0}^\infty \sum_{\path\in\Paths_n}
  \Big[A_{b^+ \path\, b^-} T_0(t;\ell_\path + 2x)
  + A_{b^- \path\, b^+} T_0(t;2 L_b + \ell_\path - 2x)\Big]\; \ud x.
\end{eqnarray}

We shall use the following lemma to simplify the bounce path term.
\begin{lemma}
  \label{lem:bounce_bgone}
  If the scattering matrix $\bS$ of a graph is $k$-independent, then
  \begin{equation*}
    \sum_{\al =1}^{2B} A_{\al \al_{n-1} \cdots \al_1 \alb}
    = J_{\al_1 \al_{n-1}} A_{\al_{n-1} \cdots \al_1},
  \end{equation*}
  where $J_{\alpha,\beta} = \delta_{\alpha,\bar{\beta}}$.
\end{lemma}

\begin{proof}
  Writing out the above,
  \begin{equation}
    \sum_{\al} A_{\al \al_{n-1} \cdots \al_1 \alb}
    = S_{\al_{n-1} \al_{n-2}} \cdots S_{\al_3 \al_2} S_{\al_2 \al_1}
    \sum_{\al} S_{\al_1 \alb} S_{\al \al_{n-1}}.
  \end{equation}
  The last sum in the above is equivalent to an element of $\bS \J \bS$,
  but $\J \bS \J \bS = \I$, see lemma \ref{lem:k_indep_S}, and $\J^{-1} =
  \J$, therefore $\bS \J \bS = \J$.
\end{proof}

Now we come to the following theorem which proves the equivalence of
this method to the trace formula method of Section~\ref{sec:trace}.
\begin{theorem}$ $
\begin{equation}
    T(t) = \frac{\cL}{\pi t} + \frac14 \tr \;(J S)
    + \sum_{n=1}^\infty \sum_{p \in \PP_n} A_p \frac{\ell_p}{n}
    \frac{t/ \pi}{t^2 +\ell_p^2}
  \end{equation}
  and consequently,
  \begin{equation}
    E_c = - \frac1{2 \pi} \sum_{n=1}^\infty \sum_{p \in \PP_n}
    \frac{A_p}{\ell_p n}.
  \end{equation}
\end{theorem}

\begin{remark}
  A special case of this proof was done for the heat kernel with
  Kirchhoff conditions by Roth \cite{Roth:FirstTrace}.  Similarly,
  Kostrykin and Schrader have an analogue of this proof (again, for
  the heat kernel) in \cite{Kostrykin:HeatKernels}.
\end{remark}

\begin{proof}
  The first term in $T(t)$ is the free space term $T_{\textrm{FS}}$,
  which we found above to be equal to $T_0(t;0)\cL = \cL/\pi t$.  Thus
  we need only consider the periodic path and bounce path terms.

  \noindent\emph{Periodic orbit contribution.}
  There is no $x$-dependence in (\ref{eq:TPO}), so the integration
  gives
  \begin{equation}
    T_{\textrm{PO}}(t) = \sum_{n=0}^\infty \sum_{\path \in \Paths_n}
    \sum_{b = 1}^B \Big[ A_{b^- \path\, b^-}
    + A_{b^+ \path\, b^+} \Big] T_0(t; \ell_{\path} + L_b) L_b.
  \end{equation}
  The two amplitudes can be written as a single sum over the directed
  bond $\alpha$,
  \begin{equation}
    \sum_{b = 1}^B \Big[ A_{b^- \path\, b^-}
    + A_{b^+ \path\, b^+}\Big] T_0(t; \ell_{\path} + L_b) L_b
    = \sum_{\alpha = 1}^{2B} A_{\alpha \path \alpha}
    T_0(t; \ell_\path + L_\alpha) L_\alpha \ .
  \end{equation}
  The path $\alpha \path \alpha$ is actually a periodic path of
  period $n+1$ which we will denote by $p$.  The amplitude of $p$ is
  $A_p = A_{\alpha \path\, \alpha}$ and its length is $\ell_p =
  \ell_\path + L_\alpha$.  Thus,
  \begin{equation}
    \label{eq:G_PO_simplify}
    T_{\textrm{PO}}(t) = \sum_{n=0}^\infty \sum_{p \in \PP_{n+1}}
    A_p \frac{\ell_p}{n+1} T_0(t; \ell_p)
    = \sum_{n=1}^\infty \sum_{p \in \PP_n}
    A_p \frac{\ell_p}{n} T_0(t; \ell_p) \ .
  \end{equation}

  \noindent\emph{Bounce path contribution.}
  We can again combine the two amplitudes and change the variables in
  the integrals to obtain
  \begin{equation}
    T_{\textrm{BP}}(t) = \frac12 \sum_{n=0}^\infty
    \sum_{\path\in\Paths_n} \sum_{\alpha=1}^{2B} A_{\al \path\, \alb}
    \int_{\ell_\path}^{ \ell_\path + 2L_\alpha} T_0(t;x) \; \ud x \ .
  \end{equation}
  We now fix $y>0$ and introduce the cutoff function,
  \begin{equation}
    H(y-x) = \cases{1 & $x \leqslant y$ \\
      0 & $x > y$ \\}  \ .
  \end{equation}
  Instead of $T_0(t;x)$ in the integral, we consider $\hat{T}_y(t;x)
  = T_0(t;x)H(y-x)$.  Since the minimum bond length $\lmin$ is greater
  than zero, taking $m$ large enough we will have
  $\ell_\path>y$ for all paths of topological length $m-1$ or
  greater.  Therefore, for a path $\path$ in $\Paths_{m-1}$ or in
  $\Paths_m$ and any $\alpha$, we can write
  \begin{equation}
    \int_{\ell_\path}^{\ell_\path + 2 L_\alpha} \hat{T}_y(t;x)\; \ud x
    = \int_{\ell_\path}^{y} \hat{T}_y(t;x)\; \ud x,
  \end{equation}
  since the integrand is identically zero on both intervals of
  integration.  We can also ignore all paths from $\Paths_n$ with $n>m$.

  We can therefore write the bounce path contribution with the given
  cutoff in the following form,
  \begin{eqnarray}
    \label{eq:three_sets}
    \hat{T}_{\textrm{BP}}(t)
    = \frac12 \sum_{n=0}^{m-2} \sum_{\path\in\Paths_n}
    \sum_{\alpha=1}^{2B} A_{\al \path\, \alb}
    \int_{\ell_\path}^{\ell_\path + 2L_\alpha} \hat{T}_y(t;x)\; \ud x \nonumber
    \\
    + \frac12 \sum_{\path\in\Paths_{m-1}}
    \sum_{\alpha=1}^{2B} A_{\al \path\, \alb}
    \int_{\ell_\path}^{y} \hat{T}_y(t;x)\; \ud x
    + \frac12 \sum_{\path\in\Paths_m}
    \sum_{\alpha=1}^{2B} A_{\al \path\, \alb}
    \int_{\ell_\path}^{y} \hat{T}_y(t;x)\; \ud x \ .
  \end{eqnarray}
  Applying Lemma~\ref{lem:bounce_bgone} to the last set of sums in
  (\ref{eq:three_sets}), we obtain
  \begin{equation}
    \sum_{\path\in\Paths_m} \sum_{\alpha=1}^{2B} A_{\al \path \alb}
    \int_{\ell_\path}^y \hat{T}_y(t;x) \; \ud x
    = \sum_{\path\in\Paths_{m-2}}
    \sum_{\beta=1}^{2B} A_{\beta \path \bar{\beta}}
    \int_{\ell_\path + 2L_\beta}^y \hat{T}_y(t;x) \; \ud x.
    \label{eq:m}
  \end{equation}
  Here the new path is the same as the old path but with the first and
  the last bond removed and $\beta$ corresponds to the last bond of
  the old path.

  Now we can take the sum corresponding to $n=m-2$ in
  (\ref{eq:three_sets}) and add it to the result of (\ref{eq:m}),
  \begin{eqnarray}
    \label{eq:joining}
    \sum_{\path\in\Paths_{m-2}}
    \sum_{\alpha=1}^{2B} A_{\al \path \alb}
    \int_{\ell_\path}^{\ell_\path + 2L_\alpha} \hat{T}_y(t;x)\; \ud x
    + \sum_{\path\in\Paths_{m-2}}
    \sum_{\beta=1}^{2B} A_{\beta \path \bar{\beta}}
    \int_{\ell_\path + 2L_\beta}^y \hat{T}_y(t;x) \; \ud x \nonumber \\
    = \sum_{\path\in\Paths_{m-2}}
    \sum_{\al =1}^{2B} A_{\al \path \alb}
    \int_{\ell_\path}^y \hat{T}_y(t;x) \; \ud x
  \end{eqnarray}
  Therefore, $\hat{T}_{\textrm{BP}}(t)$ can be rewritten exactly in the
  form of (\ref{eq:three_sets}) but with $m$ reduced by 1.  Proceeding
  by induction, we obtain
  \begin{equation}
    \hat{T}_{\textrm{BP}}(t)
    = \frac12 \sum_{\alpha}^{2B} A_{\al \alb}
    \int_0^y \hat{T}(t;x,y) \; \ud x
    + \frac12 \sum_{\alpha = 1}^{2B} J_{\al \al}
    \int_{L_\alpha}^y \hat{T}(t;x,y) \; \ud x.
  \end{equation}
  However, $J_{\al \al} = 0$ and we can take the limit
  $y \rightarrow \infty$ to get back $T_{\textrm{BP}}(t)$,
  \begin{equation}
    T_{\textrm{BP}}(t)
    = \frac12 \left[\sum_{\al=1}^{2B} (\bS \J)_{\al \al}\right]
    \int_0^\infty T_0(t;x) \; \ud x = \frac14 \tr(\bS \J)\ .
    \label{eq:G_BP_simplify}
  \end{equation}
  The significance on this term is explored thoroughly in
  \cite{Me:Index4Graphs}.  Since it is constant it vanishes upon
  differentiation and thus makes no contribution to the vacuum energy
  expression.
\end{proof}

The above method can be applied to other integral kernels, and we can also
find the vacuum energy density if we look at $-\frac 1 2
\frac{\partial}{\partial t} T_{bb}(t;x,x)$, see \cite{t:W:VEQG}.

\section{Random matrix models of vacuum energy}
\label{sec:rmt}

It has been observed by Fulling
\cite{Fulling:Private,p:SAF:GLVECOT:} that level repulsion tends to
decrease the magnitude of vacuum energy.  A natural conclusion would
be that in a chaotic system the vacuum energy should be suppressed.  In
this section we attempt to quantify and model this observation.

The first serious problem is that of comparison: vacuum energy should
be suppressed compared to what?  One cannot directly compare vacuum
energy of a chaotic system to that of an integrable one: such systems
would be too different.  Thus the right approach seems to be the
average of the energy over an appropriate ensemble of
chaotic/integrable systems

In such situations it is customary to employ random matrices as models
of chaotic systems.  Which leads to a second problem: vacuum energy is
not an exciting quantity when the spectrum is finite.  Thus, (finite)
random matrices do not immediately provide a suitable model.

In this section we use a fusion of random matrix and graph models as a
testing ground for the above conjecture.  Namely, we study quantum
graphs with equal bond lengths but with scattering matrices drawn from
the appropriate ensembles of unitary matrices.  The advantages are
clear: each individual system will have an infinite spectrum, the
spectra (in the limit of large graphs) will have the desired
statistics and the averaging can be done explicitly.

\subsection{Average vacuum energy}
\label{sec:averages}

The spectrum of a generic quantum system with a chaotic classical
counterpart is observed to behave like that of a random matrix, which
is referred to as the Bohigas-Giannoni-Schmit conjecture
\cite{p:BGS:CCQS}.  For a system with time-reversal symmetry the
appropriate ensemble of unitary matrices is the circular orthogonal
ensemble (COE) while in the absence of time-reversal symmetry it is
the circular unitary ensemble (CUE), which is the unitary group $U(N)$
with Haar measure see \cite{Haake:QuantumSignatures}.  For a system with
time-reversal symmetry and half-integer spin the random
matrix should be drawn from the circular symplectic ensemble (CSE).
To model the
vacuum energy of a generic chaotic system we consider quantum graphs
with equal bond lengths where the scattering matrix $\bS$ is taken
from an appropriate ensemble.  In drawing $\bS$ from a random matrix ensemble we must first assume that the random matrix models a graph which allows scattering between all the edges, in other words a rose graph where one central vertex connects $B$ loops.  A random matrix $\bS$ will also not, in general, correspond to scattering amplitudes accessible from matching conditions of a self-adjoint Laplace operator.  However, both these limitations seem to be intrinsic to any random matrix model.  
Our model should therefore only be regarded as a way to investigate average properties of vacuum energy when its generating spectrum has the desired random matrix level repulsion.

The average vacuum energy of such a model can be evaluated using
equation (\ref{eq:Ec_gen_graph_equal}), which expresses the vacuum
energy of a graph with equal bond lengths in terms of the eigenphases
of $\bS$.  Using the standard expression for the eigenphase density of
the random matrices results in
\begin{equation}
  \label{eq:Ec_equal_average}
  \fl \langle E_c \rangle_\beta =
  -\frac{\pi}{L} \frac{1}{\mathcal{N}_\beta}\int_0^{2\pi} \dots \int_0^{2\pi}
  \sum_{j=1}^{2B} B_2(\theta_j/2\pi) \prod_{l<m} \left|
    2 \sin \left( \frac{\theta_l - \theta_m}{2} \right) \right|^\beta \ud
  \theta_1 \dots \ud \theta_{2B} = 0
\end{equation}
where $\mathcal{N}_\beta$ is a normalization constant and
$\beta=1$ for the COE, $\beta=2$ for the CUE and $\beta=4$ for CSE.

The eigenphases of an integrable system, in contrast, behave like
uniformly distributed random numbers on the interval $[0,2\pi]$.  This
can also be modeled by integrating the vacuum energy expression,
equation (\ref{eq:Ec_gen_graph_equal}), over $2B$ independent uniform
random phases,
\begin{equation}
  \label{eq:Ec_equal_average_poisson}
  \langle E_c \rangle_{\textrm{Poisson}} =
  -\frac{\pi}{L} \frac{1}{(2\pi)^{2B}}\int_0^{2\pi} \dots \int_0^{2\pi}
  \sum_{j=1}^{2B} B_2(\theta_j/2\pi) \, \ud
  \theta_1 \dots \ud \theta_{2B}
  = 0 \ .
\end{equation}
In each case the average $\langle E_c \rangle$ of the vacuum energy is
zero.  In fact, the vacuum energy expression must be zero when
averaged over any measure which is invariant under a rotation of all
the eigenphases by some angle $\gamma$.  Indeed, in
Section~\ref{sec:equal} we have shown that
\begin{equation}
  E_c = -\frac{1}{\pi L} \sum_{n=1}^{\infty} \frac{1}{2n^2}
  \left(\tr \bS^n  +  \tr(\bS^\dagger)^n \right)\ .
\end{equation}
But for any rotationally invariant distribution of eigenphases on the
unit circle, $\langle \tr \bS^n \rangle = 0$ for all $n>0$ and thus
the mean vacuum energy is always zero.  As the Casimir force is the derivative of $E_c$ with respect to $L$ this suggests that the there
is no \textit{a priori} reason to expect either an attractive or
repulsive force based purely on the underlying nature of the classical
dynamics.

\subsection{Variance}
\label{sec:variance}

To get a handle on how the magnitude of the vacuum energy is affected
by the distribution of the eigenvalues we will calculate the variance
of the vacuum energy for the ensembles of random graphs introduced
previously.  For Poisson distributed eigenphases the variance is
\begin{eqnarray}
  \langle E_c^2 \rangle_{\textrm{Poisson}} & =
  \frac{\pi^2}{L^2} \frac{1}{(2\pi)^{2B}}
  \int_0^{2\pi} \dots \int_0^{2\pi} \sum_{j=1}^{2B} B_2^2(\theta_j/2\pi) \, \ud
  \theta_1 \dots \ud \theta_{2B} \nonumber \\ & = \frac{\pi^2 B}{90 L^2} \ ,
\end{eqnarray}
where $2B$ is the dimension of $\bS$ and we used the independence of
$\theta_j$ to conclude that
\begin{equation}
  \Big\langle  B_2(\theta_r/2\pi) B_2(\theta_j/2\pi) \Big\rangle
  = \Big\langle B_2(\theta_r/2\pi) \Big\rangle
  \Big\langle B_2(\theta_r/2\pi) \Big\rangle = 0.
\end{equation}

The variance of the vacuum energy modeled by random matrices from the
circular ensembles can be computed using expression (\ref{eq:Ec_power S}),
\begin{equation}
  \langle E_c^2 \rangle =
  \frac{1}{\pi^2L^2} \sum_{m,n=1}^\infty \frac{1}{4n^2m^2} \Big\langle
  \big( \tr \bS^n +\tr (\bS^\dagger)^n \big)
  \big( \tr \bS^m + \tr (\bS^\dagger)^m  \big) \Big\rangle.
\end{equation}
For the circular ensembles $\langle \tr \bS^n \tr (\bS^\dagger)^m
\rangle =0$ unless $m=n$ as the average of a product of matrix
elements is zero unless the number of elements of the matrix and the
number from its Hermitian conjugate are the same \cite{p:BB:DMI}.
Consequently
\begin{equation} \label{eq:varformfactor}
 \langle E_c^2 \rangle = \frac{1}{2\pi^2L^2}
 \sum_{n=1}^\infty \frac{1}{n^4} \langle |\tr \bS^n |^2 \rangle \ .
\end{equation}

We notice that $\langle |\tr\bS^n |^2 \rangle$ is the form factor of
the (finite) ensemble and use the standard formulae
\cite{Haake:QuantumSignatures} for CUE
\begin{equation}
  \langle |\tr \bS^n |^2 \rangle_{\textrm{CUE}} = \cases{(2B)^2 & $n=0$ \\
  n & $|n|  < 2B$ \\
  2B & $|n| \geqslant 2B$ \\}
\end{equation}
to obtain
\begin{equation}
  \langle E_c^2 \rangle_{\textrm{CUE}}=   \frac{1}{2\pi^2L^2}
  \left( \frac{1}{2} \Psi^{(2)}(2B) + \zeta(3)
    + \frac{B}{3} \Psi^{(3)}(2B) \right)
\end{equation}
where $\Psi^{(n)}(x)$ is the $n$-th polygamma function and $\zeta$ the
Riemann zeta function.  For all fixed $B$ this is less than
$\pi^2B/360$ the variance of the Poisson distributed eigenphases.  In
fact
\begin{equation}
\lim_{B\to \infty}  \langle E_c^2 \rangle_{\textrm{CUE}} =
  \frac{\zeta(3)}{2\pi^2L^2}   \ .
\end{equation}
This result parallels that for a random matrix model of the grand
potential considered in \cite{LebMonBoh01}.  Thus, while the variance
of the Poisson ensemble grows linearly with matrix size, the CUE
variance converges.

The relevant parts of the form factor of the COE and CSE for finite
matrix size are,
\begin{eqnarray}\label{eq:COEformfactor}
  \langle |\tr \bS^n |^2 \rangle_{\textrm{COE}}  &= \left\{ \begin{array}{lcl}
      2n - n \sum_{m=1}^n \frac{1}{m+(2B-1)/2}& & 0<n  \leqslant 2B \\
      4B - n \sum_{m=1}^{2B} \frac{1}{m+n-(2B+1)/2}& & 2B \leqslant n \\
    \end{array} \right. \\ \label{eq:CSEformfactor}
  \langle |\tr \bS^n |^2 \rangle_{\textrm{CSE}}  &= \left\{ \begin{array}{lcl}
      2n + n \sum_{m=1}^{n} \frac{1}{(2B+1)/2-m}& & 0 < n  \leqslant 2B \\
      4B & & 2B \leqslant n
    \end{array} \right.
\end{eqnarray}
Note that in the CSE form factor the double degeneracy of the
eigenphases of $\bS$ (Kramers' degeneracy) has not been lifted.  Using
(\ref{eq:varformfactor}) we evaluate
\begin{eqnarray}
  \fl \langle E_c^2 \rangle_{\textrm{COE}}= \frac{1}{\pi^2L^2} \left(
    \left( \frac{1}{2} \Psi^{(2)}(2B) + \zeta(3) \right) \left(
      1+\frac{1}{2}\Psi(B+1/2) \right)
    -\sum_{n=1}^{2B-1} \frac{\Psi(n+B+1/2)}{2n^3} \right. \\
  \left. + \sum_{n=2B}^{\infty} \left[\frac{2B}{n^4} +
      \frac{\Psi(n-B+1/2) - \Psi(n+B-1/2)}{2n^3} \right] \right)
  \label{eq:COEvar}
\end{eqnarray}
\begin{eqnarray}
\fl  \langle E_c^2 \rangle_{\textrm{CSE}}= \frac{1}{\pi^2L^2} \left(
    \left( \frac{1}{2} \Psi^{(2)}(2B) + \zeta(3) \right) \left(
      1+\frac{1}{2}\Psi(1/2-B) \right)
  \right. \\
  \left. -\sum_{n=1}^{2B-1} \frac{\Psi(n-B+1/2)}{2n^3} +
    \frac{B}{3}\Psi^{(3)}(2B) \right) \label{eq:CSEvar}
\end{eqnarray}
For $B>1$ $\langle E_c^2 \rangle_{\textrm{COE}}$ and $\langle E_c^2
\rangle_{\textrm{CSE}}$ are less than $\langle E_c^2
\rangle_{\textrm{Poisson}}$.

While the results do not have a simple closed form, in the limit of
large matrices the result is rather concise,
\begin{equation}
   \lim_{B\to \infty}  \langle E_c^2 \rangle_{\textrm{COE}}
   = \frac{\zeta(3)}{\pi^2L^2} =
   \lim_{B\to \infty}  \langle E_c^2 \rangle_{\textrm{CSE}}
\end{equation}

Modeling the vacuum energy variance of a quantum graph through random
matrices suggests that the magnitude of the vacuum energy where
eigenphases of $\bS$ experience level repulsion are indeed smaller on
average than those where the eiegenphases are Poisson distributed.
Moreover, the magnitude gets smaller as the level repulsion increases
from linear (COE) to quadratic (CUE) and quartic (CSE).  Indeed, to
compare the effect of the increased level repulsion in CSE we need to
lift the Kramers' degeneracy: otherwise the degeneracy ``compensates''
the repulsion.  Without the degeneracy, the result for CSE becomes 4
times smaller:
\begin{equation*}
  \lim_{B\to \infty}  \langle E_c^2 \rangle_{\textrm{CSE/Kramers}}
  = \frac{\zeta(3)}{4\pi^2L^2},
\end{equation*}
thus leading to
\begin{equation*}
  \langle E_c^2 \rangle_{\textrm{COE}} > \langle E_c^2
  \rangle_{\textrm{CUE}} > \langle E_c^2 \rangle_{\textrm{CSE/Kramers}}
\end{equation*}
for $B>1$.


\section{Conclusions}

Through both the method of images and the trace formula, we
demonstrate that the vacuum energy in quantum graphs is a well-defined
quantity (i.e.\ it is both convergent and a smooth function of the
bond lengths).  The closed form expression (\ref{eq:answer}) is
dependent only on the periodic paths in the quantum graph; this is a
consequence of the exactness of the trace formula which includes only
those paths.  Having demonstrated how the bounce paths (closed paths that
are not periodic) cancel when using the method of images we hope
that our proof will shed light on the observation that periodic paths
provide the correct leading asymptotic behavior even when the trace
formula is only semiclassically correct.

The smoothness of the expression for the vacuum energy in a quantum
graph allows us to suggest an alternative method for its calculation,
by approximating with systems with simpler geometries.  In the case of
graphs the ``simpler geometry'' means rational bond lengths, where an
explicit expression for the vacuum energy is obtained.

We also suggest a random ensemble model for the vacuum energy when the
statistics of the spectrum are Poisson (for integrable systems) or
random matrix (for chaotic systems).  We find the average energy in
both cases to be zero, thus giving no \emph{a priori} reason to expect
a positive or negative energy from the dynamics.  Furthermore, we find
the variance of the energy and conclude that the magnitude of the energy
is typically smaller when the level repulsion is stronger.  We stress
that this prediction is only correct in the probabilistic sense and no
conclusions about particular systems can yet (if ever!) be drawn.



\ack
The authors would like to thank S.A.~Fulling, J.P.~Keating, K.~Kirsten,
P.~Kuchment, M.~Pivarsky and B.~Winn for their helpful comments.  This
material is based upon work supported by the National Science
Foundation under Grants No.~DMS-0604859, PHY-0554849 and DMS-0648786.
The authors would also like to thank the Isaac Newton Institute for
Mathematical Sciences, Cambridge, UK, where part of the research took
place.


\section*{References}
\bibliographystyle{unsrt}
\bibliography{bhw}

\end{document}